\documentclass[11pt]{article}

\usepackage{amsmath,amsthm}

\newtheorem{lemma}{Lemma}

\newtheorem{corollary}{Corollary}

\usepackage{amssymb,calc}
\usepackage{mdframed}
\usepackage{microtype}
\usepackage{todonotes}
\usepackage{xspace}
\usepackage[pdfborder={0 0 0},colorlinks=true,linkcolor=blue]{hyperref}
\usepackage{varioref}

\usepackage{thmtools}
\usepackage{thm-restate}

\usepackage{algorithm} 
\usepackage[noend]{algpseudocode} 

\usepackage{vmargin}
\setmarginsrb{1in}{1in}{1in}{1in}{0mm}{0mm}{0mm}{7mm}

\newcommand{\dfvs}{\textsf{DFVS}\xspace}
\newcommand{\tfvs}{\textsf{TFVS}\xspace}

\newcommand{\true}{\textbf{True}\xspace} 
\newcommand{\false}{\textbf{False}\xspace}

\newcommand{\pInT}{\ensuremath{p\_is\_in\_a\_triangle}}
\newcommand{\AAA}{\ensuremath{\mathcal{A}}\xspace}

\newcommand{\defproblem}[3]{
	\vspace{0.5em}
	\noindent\fbox{
		\begin{minipage}{0.97\textwidth}
			\begin{tabular*}{\textwidth}{@{\extracolsep{\fill}}lr} #1 \\ \end{tabular*}
			{\bf{Input:}} #2  \\
			{\bf{Output:}} #3
		\end{minipage}
	}\vspace{0.5em}}

\begin{document}

\title{A $2$-Approximation Algorithm for Feedback Vertex Set in Tournaments}

\author{
Daniel Lokshtanov\thanks{University of California, Santa Barbara, USA. \texttt{daniello@ucsb.edu}}
\and Pranabendu Misra\thanks{University of Bergen, Bergen, Norway. \texttt{pranabendu.misra@uib.no}}
\and Joydeep Mukherjee\thanks{Indian Statistical Institute, Kolkata, India. \texttt{joydeep.m1981@gmail.com}}
\and Fahad Panolan\thanks{University of Bergen, Bergen, Norway. \texttt{fahad.panolan@uib.no}}
\and Geevarghese Philip\thanks{Chennai Mathematical Institute, India. \texttt{gphilip@cmi.ac.in}}
\and Saket Saurabh\thanks{The Institute of Mathematical Sciences, HBNI, Chennai, India. \texttt{saket@imsc.res.in}}
}

\date{}

\maketitle

\begin{abstract}
A {\em tournament} is a directed graph $T$ such that every pair of
vertices is connected by an arc. A {\em feedback vertex set} is a
set $S$ of vertices in $T$ such that $T - S$ is acyclic. We
consider the {\sc Feedback Vertex Set} problem in
tournaments. Here the input is a tournament $T$ and a weight
function $w : V(T) \rightarrow \mathbb{N}$ and the task is to find
a feedback vertex set $S$ in $T$ minimizing $w(S) = \sum_{v \in S}
w(v)$. We give the first polynomial time factor $2$ approximation
algorithm for this problem. Assuming the Unique Games conjecture,
this is the best possible approximation ratio achievable in
polynomial time. \end{abstract}
\section{Introduction}
A {\em feedback vertex set} (FVS) in a graph $G$ is a vertex subset $S$ such that $G - S$ is acyclic. 
In the case of directed graphs, it means $G - S$ is a directed acyclic graph (DAG). 
In the {\sc (Directed) Feedback Vertex Set} (\textsf{(D)FVS})
problem we are given as input a (directed) graph $G$ and a weight
function $w : V(G) \rightarrow \mathbb{N}$.  The objective is to
find a minimum weight feedback vertex set $S$.
Both the directed and undirected version of the problem are NP-complete~\cite{GJ79} and have been extensively studied from the perspective of approximation algorithms~\cite{BafnaBF99,EvenNSS98}, parameterized algorithms~\cite{ChenLLOR08,CyganNPPRW11,KociumakaP14}, exact exponential time algorithms~\cite{Razgon07,XiaoN15} as well as graph theory~\cite{erdHos1965independent,reed1996packing}.

In this paper we consider a restriction of \dfvs, namely the {\sc Feedback Vertex Set in Tournaments} (\tfvs) problem, from the perspective of approximation algorithms. A {\em tournament} is a directed graph $G$ such that every pair of vertices is connected by an arc, and \tfvs\ is simply \dfvs when the input graph is required to be a tournament. 
We refer to the textbook of Williamson and Shmoys~\cite{willshmoys_book} for an introduction to approximation algorithms. 
Even this restricted variant \dfvs has applications in voting systems and rank aggregation
and is quite well-studied~\cite{CaiDZ00,Dom201076,GaspersM13,RamanS06,mnich20167,KumarL16}.
It is formally defined as follows.

\defproblem{\sc Feedback Vertex Set in Tournaments (TFVS)}{A tournament $G$ and a weight function $w:V(G) \rightarrow {\mathbb N}$.}{A minimum weight FVS of $G$.}

The problem has several simple $3$-approximation algorithms. It is
well known that a tournament has a directed triangle if and only
if there is a directed triangle~\cite{Dom201076}. Then a
$3$-approximation solution for the \emph{unweighted}
version\footnote{Where all the vertices have the same weight.} of
\tfvs\ is easily constructed as follows. If there is a directed
triangle in the tournament put all the vertices of the triangle in
the solution and delete them from the tournament. We repeat the
above process until the tournament becomes triangle
free\footnote{This will not, in general, give a $3$-approximation
  for a \emph{weighted} instance.}.  Another simple
$3$-approximation algorithm for \tfvs\ is given
in~\cite{bar2005}. The first algorithm with a better approximation
ratio was given by Cai et al.~\cite{CaiDZ00}, who gave a
$5/2$-approximation algorithm using the local ratio technique of
Bar-Yehuda and Even~\cite{BARYEHUDA}.  Recently, Mnich et
al.~\cite{mnich20167} gave a $7/3$-approximation algorithm using
the iterative rounding technique. They observe that the
approximation-preserving reduction from {\sc Vertex Cover} to
\tfvs of Speckenmeyer~\cite{speck_feedback} implies that, assuming
the Unique Games Conjecture (UGC)~\cite{khot2008vertex}, \tfvs
cannot have an approximation algorithm with factor smaller than
$2$.  The more general \dfvs problem has a
factor-$O(\min\{\log n \log \log n, \log \tau \log \log \tau \})$
approximation ~\cite{Seymour95,even2000divide} where \(n\) is the
number of vertices in the input tournament and $\tau$ is the cost
of an optimal solution, and it is known that \dfvs cannot have a
factor-$\alpha$ approximation for any constant $\alpha > 1$ under
the
UGC~\cite{guruswami2016simple,guruswami2011beating,svensson2012hardness}.
A related problem is {\sc $3$-Hitting Set} or {\sc Vertex Cover}
in $3$-uniform hypergraphs.  Here the input is a universe $U$ and
a family $\mathbb F$ of subsets of $U$ of size at most $3$.  The
goal is to find a minimum subset $S$ of the universe that
intersects every set in $\mathbb F$.  Observe that \tfvs is a
special case of this problem, since \tfvs reduces to hitting all
the directed triangles in the tournament.  While it is NP-hard to
approximate {\sc $3$-Hitting Set} better than factor
$2$~\cite{dinur2005new}, under the UGC there can be no polynomial
time approximation better than factor
$3$~\cite{khot2008vertex}\footnote{These results actually hold for
  the more general problem of {\sc Vertex Cover} in $k$-uniform
  hypergraphs.}.
Mnich et al.~\cite{mnich20167} state that their algorithm ``{\em gives hope that a $2$-approximation algorithm, that would be optimal
under the UGC, might be achievable} (for \tfvs)''. 
In this paper we show that this is indeed the case, by giving a (randomized) $2$-approximation algorithm for \tfvs. More formally, we prove the following theorem.

\begin{restatable}{thm}{mainthm}
\label{thm:main}
There exists a randomized algorithm that, given a tournament $G$
on \(n\) vertices and a weight function $w$ on \(G\), runs in time
$O(n^{34})$ and outputs a feedback vertex set $S$ of $G$. With
probability at least $1/2$, $S$ is a $2$-approximate solution of
$(G,w)$.
\end{restatable}

This algorithm can be easily \emph{derandomized} in quasi-polynomial time.

\paragraph{Our Methods.} Our algorithm is inspired by the methods
and analysis of Fixed Parameter Tractable (FPT)-algorithms.  A
well known technique in FPT algorithm is \emph{branching} where we
try to guess if a vertex is in the optimal solution or not.
Similarly, our approximation algorithm tries to randomly sample a
vertex $p$ of the tournament $(i)$ which is \emph{not} contained
in some optimal solution, and $(ii)$ whose in-degree and
out-degree are each at most a constant fraction of $n$.  Assuming
that the size of the optimal solution is upper bounded by a
constant fraction of $n$, the random sampling succeeds with a
constant probability\footnote{When the size of the optimal
  solution is large, the algorithm picks a constant fraction of
  lowest weight vertices into the approximation solution, to
  obtain the reduced instance.}.  With the vertex $p$ in hand, we
reduce the input instance into smaller instances, defined by the
in-neighborhood and the out-neighborhood of $p$, which are then
solved recursively. By the the properties of $p$, the cardinality
of the vertex set of each of these instances is upper-bounded by a
constant fraction of $n$.  This step is reminiscent of reduction
rules that are frequently applied in FPT algorithms and
kernelization.  We show that we can recover a $2$-approximation
for the input instance from $2$-approximate solutions of the
reduced instances, with a constant probability of success.  By
repeated application, this process gradually decomposes the input
instance into a collection of constant size instances, which are
then solved by brute force.  This leads to a 
$2$-approximation algorithm for \tfvs which runs in randomized
polynomial time.  We believe that the connection to FPT algorithms
and analysis is a key feature of our algorithm, which will be
applicable for other problems.

\section{Preliminaries}\label{pre}

In this paper we work with directed graphs (or \emph{digraphs})
that do not contain any self loops or parallel arcs. 
We use $V(G)$ to denote the vertex set of a
digraph \(G\) and $E(G)$ to denote the set of arcs of
$G$.  We use the notation $uv$ to denote 
an arc from vertex $u$ to vertex $v$ in a
digraph. Vertices \(u,v\) are \emph{incident with} arc \(uv\).  A \emph{tournament} is a digraph in which
there is exactly one arc between any two vertices. The set of
\emph{out-neighbors} of a vertex $v$ in a digraph \(G\) is defined
to be $N^{+}(v):=\{u\;\mid\;vu\in E(G)\}$, and the set of
\emph{in-neighbors} of $v$ in \(G\) is defined to be
$N^{-}(v):=\{u\;\mid\;uv\in E(G)\}$. For an integer
\(\ell\geq{}3\) a \emph{directed cycle of length \(\ell\)} in a
digraph \(G\) is an alternating sequence
\(C=v_{1}a_{1}v_{2}a_{2}\dots{}v_{\ell}a_{\ell}\) where
\(\{v_{1}\dots,v_{\ell}\}\subseteq{}V(G)\) is a set of \(\ell\)
distinct vertices of \(G\) and
\(\{a_{1}\dots,a_{\ell}\}\subseteq{}E(G)\) is a subset of arcs of
\(G\) where \(a_{i}=v_{i}v_{i+1};1\leq{}i<\ell\) and
\(a_{\ell}=v_{\ell}v_{1}\). A digraph is \emph{acyclic} if it does
not contain a directed cycle. A \emph{triangle} in a digraph is a directed cycle of length three. In this paper we use
the term ``triangle'' exclusively to denote directed triangles.  A
\emph{topological sort} of a digraph $G$ with \(n\)
vertices is a permutation $\pi:V(G)\mapsto [n]$ of the vertices of
the digraph such that for all arcs $uv\in E(G)$, it is the case that $\pi(u)<\pi(v)$. Such a permutation exists for a digraph
\(G\) if and only if \(G\) is
acyclic~\cite{gutinDigraphsBook}. For an acyclic tournament, the
topological sort is
unique~\cite{gutinDigraphsBook}. \emph{Deleting} a vertex \(v\)
from digraph \(G\) involves removing, from \(G\), the vertex \(v\)
and all those arcs in \(G\) with which \(v\) is
incident in \(G\). We use $G-v$ to denote the digraph obtained by
deleting a vertex \(v\in{}V(G)\) from digraph \(G\). For a vertex
set $S\subseteq{}V(G)$ we use $G - S$ to denotes the digraph
obtained from digraph $G$ by deleting all the vertices of $S$.

A {\em feedback vertex set} (FVS) of a digraph \(G\) is a vertex set $S$ such that $G - S$ is acyclic. A vertex set is a {\em feasible
  solution} if and only if it is an FVS. Given a weight function
$w : V(G) \rightarrow \mathbb{N}$ the {\em weight} of a vertex set
$S$ is $w(S) = \sum_{v \in S} w(v)$. An FVS $S_{OPT}$ of $G$ is an {\em optimal} solution of the instance $(G, w)$ if every other FVS $S$ of $G$ satisfies $w(S) \geq w(S_{OPT})$. 
A FVS $S$ of $G$ is called $2$-{\em approximate solution} of the instance $(G, w)$ if $w(S) \leq 2w(S_{OPT})$ for an optimal solution $S_{OPT}$ of $(G, w)$.
An FVS $S$ is called $p$-disjoint for a vertex $p$ if $p \notin S$, and further, $S$ is said to be an {\em optimal $p$-disjoint FVS of} $(G,w)$ if, for every $p$-disjoint solution $S'$ we have $w(S') \geq w(S)$. Note that an optimal $p$-disjoint solution of $(G,w)$ is not necessarily an optimal solution of $(G, w)$. On the other hand if an optimal solution $S_{OPT}$ of $(G, w)$ happens to be $p$-disjoint then $S_{OPT}$ is also an optimal p-disjoint solution of $G$. A $p$-disjoint FVS $S$ of $G$ is called $2$-{\em approximate $p$-disjoint solution} of the instance $(G, w)$ if $w(S) \leq 2w(S')$ for an optimal $p$-disjoint solution $S'$ of $(G, w)$.

In the following we will assume that $G$ is a
tournament on $n$ vertices, and $w : V(G) \rightarrow \mathbb{N}$
is a weight function. Furthermore, for any induced subgraph $H$ of $G$, we assume that $w$ defines a weight function, when restricted to $V(H)$. 
We will frequently make use of the following lemma which  directly follows from the fact that acyclic digraphs are closed under vertex deletions.

\begin{lemma}\label{lem:hereditarySolutions}
  Let \(S\) be an FVS of a digraph \(G\) and let \(X\) be a subset
  of the vertex set of \(G\). Then \(S\setminus{}X\) is an FVS of
  the digraph \(G-X\). If \(S^{\star}\) is an \emph{optimal}
  solution of an instance \((G,w)\) of \tfvs and \(X\) is a subset
  of \(S^{\star}\) then \(S^{\star}\setminus{}X\) is an
  \emph{optimal} solution of the instance \(((G-X),w)\), of weight
  \(w(S^{\star})-w(X)\).
\end{lemma}

We use the following lemma to prove the correctness  our algorithm in the later section.

\begin{lemma}\label{lem:InOutNbrhoodSplit}
  Let \((G,w)\) be an instance of \tfvs. 
  \begin{itemize}
  \item[$(i)$] A vertex \(v\in{}G\) is not part of any triangle in \(G\)
    if and only if every arc between a vertex in \(N^{-}(v)\) and
    a vertex in \(N^{+}(v)\) is of the form
    \(xy\;;\;x\in{}N^{-}(v),y\in{}N^{+}(v)\).

  \item[$(ii)$] Let \(x\in{}V(G)\) be a vertex which is not part of any
    triangle in \(G\). Let \(H_{in}=G[N^{-}(x)]\) and
    \(H_{out}=G[N^{+}(x)]\) be the subgraphs induced in \(G\) by
    the in- and out-neighborhoods of vertex \(x\),
    respectively. A set \(S\) is an FVS of digraph \(G\) if and only
    if \(S\cap{}V(H_{in})\) is an FVS of the subgraph \(H_{in}\)
    and \(S\cap{}V(H_{out})\) is an FVS of the subgraph
    \(H_{out}\).
  \end{itemize}
\end{lemma}

\begin{proof}
   
  Suppose vertex \(v\) is not part of any triangle in \(G\). If
  there is an arc \(st\) in \(G\) where vertex \(s\) is in the
  out-neighborhood \(N^{+}(v)\) of vertex \(v\) and vertex \(t\)
  is in its in-neighborhood \(N^{-}(v)\) then the vertices
  \(\{s,v,t\}\) form a triangle containing vertex \(v\), a
  contradiction. So every arc between vertices \(x\in{}N^{-}(v)\)
  and \(y\in{}N^{+}(v)\) is directed from \(x\) to
  \(y\). Conversely, if vertices \(\{v,s,t\}\) form a triangle
  and---without loss of generality---\(vs\) is an arc in \(G\)
  then we have that both \(st\) and \(tv\) are arcs in \(G\). Thus
  \(s\in{}N^{+}(v),t\in{}N^{-}(v)\), and arc \(st\) is not of the
  form \(xy\;;\;x\in{}N^{-}(v),y\in{}N^{+}(v)\).

Now prove statement $(ii)$ of the lemma. Let $S$ be an FVS of $G$. As $H_{in}-(S\cap V(H_{in}))$ and $H_{out}-(S\cap V(H_{out}))$ are subgraphs of  $G-S$ (which is a DAG), we have that  \(S\cap{}V(H_{in})\) is an FVS of  \(H_{in}\)   and \(S\cap{}V(H_{out})\) is an FVS of    \(H_{out}\). 
Now we prove the other direction. Let $S\subseteq V(G)$ be such that \(S\cap{}V(H_{in})\) is an FVS of  \(H_{in}\)   and \(S\cap{}V(H_{out})\) is an FVS of    \(H_{out}\). Since $H_{in}-S$ is an acyclic tournament, there is a unique topological sort $u_1,\ldots,u_{\ell}$ of $H_{in}-S$,  where $\{u_1,\ldots,u_{\ell}\}=V(H_{in})\setminus S$. Also, since  $H_{out}-S$ is an acyclic tournament, there is a unique topological sort $v_1,\ldots,v_{\ell'}$ of $H_{out}-S$, where $\{v_1,\ldots,v_{\ell'}\}=V(H_{out})\setminus S$. Since $x$ is not part of a triangle in $G$, by statement $(i)$ of the lemma, there is no arc from  a vertex in $\{v_1,\ldots,v_{\ell'}\}$ to a vertex in $\{u_1,\ldots,u_{\ell}\}$. This implies that $u_1,\ldots,u_{\ell},x,v_1,\ldots,v_{\ell'}$ is a topological sort of $G-S$. Therefore $S$ is an FVS of $G$. 
\end{proof}

\section{The Algorithm}\label{algo}

We begin with an informal overview. Let \(G\) be a digraph and
\(w:V(G)\to\mathbb{N}\) be a weight function on the vertices of
\(G\). If \(S\) is an \emph{optimal} FVS for the instance
\((G,w)\) and \(v\) is a vertex in \(S\) then
(\autoref{lem:hereditarySolutions}) \(S\setminus{}\{v\}\) is an
optimal FVS of the instance \((G-v, w)\), and its weight is
exactly \(w(S)-w(v)\). Note that this need not be the case for
vertices \emph{outside} of \(S\); deleting a vertex \(x\notin{}S\)
may not bring down the weight of an optimal FVS. As a simple
example, consider the tournament on four vertices \(\{a,b,c,x\}\)
where (i) \(\{a,b,c\}\) form a triangle, (ii) vertex \(x\) has
in-degree three, and (iii) all vertices have weight one. An
optimum FVS of this instance consists of any one of the three
vertices \(\{a,b,c\}\) and has weight one. An optimum FVS of the
digraph \(G-x\) is also of this same form, and has weight one as
well.

Thus if we are given the promise that a vertex \(v\) is in some
optimal FVS of \((G,w)\) then we can safely delete \(v\) from
\(G\) and recursively find an optimal FVS \(S'\) of the smaller
instance \((G-v,w)\), to get an optimal FVS \(S'\cup\{v\}\) of the
original instance \((G,w)\). If we don't know that vertex \(v\) is
in some optimal FVS of \((G,w)\) then we cannot safely make such a
reduction.

It turns out that if we are willing to accept the lesser promise
of ``half a vertex'' being in an optimal solution then we
\emph{can} safely make an analogous reduction which preserves a
\emph{2-approximate} solution for the \tfvs\ instance. More precisely, suppose we are given a pair of vertices \(u,v\in{}V(G)\;;\;w(v)\leq{}w(u)\)
and the promise that some optimal solution contains \emph{at
  least one} out of \(\{u,v\}\). Then---see
\autoref{lem:pickLighterVertex} (with an assumption that there is an optimal solution not containing $p$)---vertex \(v\) must belong to some
2-approximate solution for the instance \((G,w)\). Indeed, if we
delete \(v\) from \(G\) and reduce the weight of vertex \(u\) by
\(w(v)\) to get a smaller instance, then for \emph{any}
2-approximate solution \(S'\) of this smaller instance, the set
\(S'\cup\{v\}\) is a 2-approximate solution of the original instance
\((G,w)\).

So to find a 2-approximate solution for \tfvs it is enough
to---repeatedly---find pairs of vertices with the guarantee that
there is an optimal solution which contains at least one of these
two vertices. For this we use the observation that a tournament
contains a directed cycle if and only if it contains a directed
triangle. Let \(G\) be a tournament and \(\{u,v,x\}\) the vertex
set of a directed triangle in \(G\). If there is an optimal
solution which does \emph{not} contain vertex \(x\) then
\(\{u,v\}\) is a pair of vertices with the required property. So
it is enough to be able to repeatedly find a vertex which (i)
belongs to a directed triangle, and (ii) is not part of some
optimal solution. Call a vertex which has these two properties, an
``unimportant'' vertex.

If we could consistently find an unimportant vertex with some good
probability then we could solve the problem with a good
probability of success. One way to do this would be
to---somehow---ensure that a constant fraction---say, \(1/3\)---of
the entire vertex set is unimportant; a vertex picked uniformly at
random would then be unimportant with probability \(1/3\). So the
``bad case'' is when only a very small part of the vertex set is
unimportant; equivalently, when a large fraction of the vertex
set---here, \(2/3\)---is part of \emph{every} optimal
solution. This in turn implies that there \emph{is} an optimal
solution which contains a large fraction---\(2/3\)---of the vertex
set. If we can---somehow---process those cases where there is an
optimal solution which contains a very large fraction of the
vertex set then we will be able to consistently find unimportant
vertices with good probability.

Let \(S\) be an optimal solution which contains more than \(2/3\)
of the vertex set of \(G\). Consider the set \(L\) of the
\(|V(G)|/6\) vertices of the \emph{smallest} weight in \(G\). Then
the weight of the vertex set \(L\) is at most a quarter
(\(=\frac{1}{6}/\frac{2}{3}\)) of the weight of the optimum
\(S\). This suggests that picking all of \(L\) into a solution
should not result in a solution which is heavier than the optimum
by a factor of \(5/4\). Indeed, something stronger holds for
2-approximate solutions. We show---see
\autoref{lem:largeSoln}---that there is a 2-approximate solution
which contains \emph{all} of \(L\). Indeed, we can delete \(L\)
from \(G\) and modify the weights of the remaining vertices in a
certain way to get an instance \((G-L, w')\) such that for
\emph{any} 2-approximate solution \(S'\) of \((G-L, w')\), the set
\(L\cup{}S'\) is a 2-approximate solution for the original instance
\((G, w)\).

We now give a high level conceptual sketch of the algorithm,
hiding some details required for getting good bounds on the
running time and success probability.  Our algorithm has two
phases. In each phase it computes a feasible solution, and at the end it returns the solution of smaller weight among these two. We prove---along the lines suggested by the above discussion---that at least one of these solutions must be a 2-approximate solution. Recall that \((G,w)\) denotes the input instance where \(G\) has \(n\) vertices.

\textbf{Phase~1} of the algorithm computes a \emph{candidate}
2-approximate solution \(A_{1}\) for \((G,w)\) \emph{assuming}
that there is an \emph{optimum} solution \(S\) with
\(|S|\geq{}\frac{2n}{3}\). To do this the algorithm deletes the
set \(L\) of the \(n/6\) vertices of the \emph{smallest} weight in
\(G\), modifies the weights of the remaining vertices in as
specified in \autoref{lem:largeSoln}, and recursively finds a
2-approximate solution \(B_{1}\) of the resulting instance
\((G-L,w')\). The candidate 2-approximate solution from this step
is \(A_{1}=L\cup{}B_{1}\).

\textbf{Phase~2} of the algorithm computes another candidate
2-approximate solution \(A_{2}\) for \((G,w)\) assuming that
\emph{no} optimum solution has \(2n/3\) or more vertices. To do
this the algorithm picks a ``pivot'' vertex \(p\) at random. 
If \(p\) is \emph{not} part of any triangle in \(G\) then the
algorithm recursively finds 2-approximate solutions
\(S_{1},S_{2}\) of the subgraphs \(H_{in}\) and \(H_{out}\)
induced by the in- and out-neighborhoods of vertex \(p\),
respectively, and sets the candidate 2-approximate solution from
this phase to be \(A_{2}=S_{1}\cup{}S_{2}\). This is safe by
\autoref{lem:InOutNbrhoodSplit}.

If the pivot vertex \(p\) \emph{is} part of some triangle in \(G\)
then the algorithm assumes that \(p\) is unimportant, and applies a
\emph{reduction procedure} to obtain an instance where vertex
\(p\) is not in any triangle. This procedure chooses two vertices
\(\{u,v\}\;;\;w(v)\leq{}w(u)\) which form a triangle together with
\(p\). It then deletes \(v\) from \(G\) and modifies\footnote{See
  \autoref{lem:largeSoln} for the specifics.} the weight of \(u\)
to get a new instance \((G-v,w')\). The reduction procedure
consists of the repeated application of this step as long as the
pivot vertex \(p\) is part of some triangle, and stops when it
obtains a subgraph \(H\) in which vertex \(p\) is not part of any
triangle. Now the algorithm recurses on the in and
out-neighborhoods of \(p\) in digraph \(H\) as described in the
previous paragraph, to get a 2-approximate solution \(B_{2}\). The
candidate 2-approximate solution from this phase is
\(A_{2}=D\cup{}B_{2}\) where \(D\) is the set of all vertices
\(v\) deleted from \(G\) by the reduction step to get to the digraph
\(H\).  If \(w(A_{1})<w(A_{2})\) then the algorithm outputs
\(A_{1}\); otherwise it outputs \(A_{2}\).

To prove that this recursive procedure runs in polynomial time we
need to ensure that neither of the digraphs \(H_{in},H_{out}\) in the
recursive step is ``too small''; more specifically, that the
number of vertices in each of \(H_{in},H_{out}\) is
\emph{upper}-bounded by a fraction of the number of vertices \(n\)
in the digraph \(G\) given as input to Phase 2. We enforce this by
picking the pivot vertex \(p\) from among those vertices of \(G\)
whose in- and out-degrees are upper-bounded by a certain fraction
of \(n\).

In the rest of this section we give a more formal description of
the algorithm, prove its correctness, and show that it runs in
polynomial time. We begin by proving a couple of lemmas which
formalize some ideas from the above discussion. Our first lemma
pertains to the case when there is an optimal solution which
contains a large fraction of the vertex set.

\begin{lemma}\label{lem:largeSoln}
  Let \((G,w)\) be an instance of \tfvs where \(G\) has \(n\)
  vertices, and which has an optimal solution \(S^{\star}\) that
  contains at least \(2n/3\) vertices of \(G\).  Let
  $D \subseteq V(G)$ be a set of $\frac{n}{6}$ vertices of the
  smallest weight in $V(G)$, ties broken arbitrarily, and let
  $\Delta = \max_{v \in D} w(v)$ be the weight of the heaviest
  vertex in \(D\). 
Let \(w'\colon V(G)\setminus D \rightarrow {\mathbb N}\) be the weight function  which assigns the weight \(w(v)-\Delta\) to each vertex \(v\) of \(G-D\).
If $R_{approx}$ is a $2$-approximate solution of the reduced
instance $(G - D, w')$ then $R_{approx} \cup D$ is a
$2$-approximate solution of the instance $(G, w)$.
\end{lemma}

\begin{proof}
  Let \(R^{\star}\) be an optimum solution of the reduced instance
  $(G - D, w')$. Then \(w'(R_{approx})\leq{}2w'(R^{\star})\). 
  From \autoref{lem:hereditarySolutions} we get that
  \(S^{\star}\setminus{}D\) is a---not necessarily
  optimal---solution of the reduced instance \((G - D,
  w')\). Since \(R^{\star}\) is an optimum solution of this
  instance we have that
  \(w'(S^{\star}\setminus{}D)\geq{}w'(R^{\star})\). Since
  \(w'(v)=(w(v)-\Delta)\) holds for each vertex
  \(v\in{}(S^{\star}\setminus{}D)\) we get that
  \(w'(S^{\star}\setminus{}D) = (w(S^\star \setminus D)-|S^{\star}\setminus{}D|\cdot{}\Delta) \leq(w(S^{\star})-|S^{\star}\setminus{}D|\cdot{}\Delta)\). Since
  \(|S^{\star}\setminus{}D|\geq{}(\frac{2n}{3}-\frac{n}{6})=\frac{n}{2}\)
  we get that
  \(w'(S^{\star}\setminus{}D)\leq(w(S^{\star})-\frac{\Delta\cdot{}n}{2})\).
  Hence
  \(w'(R^{\star})\leq{}w'(S^{\star}\setminus{}D)\leq{}(w(S^{\star})-\frac{\Delta\cdot{}n}{2})\). 

  Thus
  \(w'(R_{approx})\leq{}2w'(R^{\star})\leq(2w(S^{\star})-\Delta\cdot{}n)\). Since
  the set \(R_{approx}\) is disjoint from the deleted set \(D\) we
  have that \(w'(v)=w(v)-\Delta\) holds for each vertex
  \(v\in{}R_{approx}\). Hence
  \(w(R_{approx}) = w'(R_{approx})+|R_{approx}|\cdot{}\Delta \leq
  (2w(S^{\star})-\Delta\cdot{}n)+|R_{approx}|\cdot{}\Delta =
  (2w(S^{\star})-\Delta(n-|R_{approx}|))\).
  Since \(w(v) \leq \Delta\) holds for each vertex \(v\in{}D\) we
  have that \(w(D)\leq{}|D|\cdot{}\Delta\). 
Hence
\begin{eqnarray*}
 w(R_{approx}\cup{}D) &=& w(R_{approx})+w(D)\\
& \leq&  (2w(S^{\star})-\Delta(n-|R_{approx}|)+|D|\cdot{}\Delta)\\
& =&  (2w(S^{\star})-\Delta(n-|R_{approx}|-|D|))\\
& =&  (2w(S^{\star})-\Delta(n-|R_{approx}\cup{}D|))\\
& \leq&  2w(S^{\star}).
\end{eqnarray*}
Here the last inequality follows from the
  fact that \(|R_{approx}\cup{}D| \leq n = |V(G)|\).
\end{proof}

The next lemma shows that given $\{p,u,v\}$, we can safely pick a lighter weight vertex of the two vertices $u$ and $v$ 
into a  2-approximate $p$-disjoint solution.

\begin{lemma}\label{lem:pickLighterVertex}
  Let   \((G,w)\) be an instance of \tfvs\ and $p\in V(G)$.  Let \(\{u,v\}\) be two vertices such that
  (i) \(\{p, u,v\}\) form a triangle in $G$, and
  (ii) \(w(v)\leq{}w(u)\). Let \(w'\) be the weight function
  defined by: $(a)$ \(w'(v)=0\) ,$(b)$ \(w'(u)=w(u)-w(v)\), and $(c)$
  \(w'(x)=w(x)\) for all vertices \(x\notin{}\{u,v\}\). 
Then for every  $2$-approximate $p$-disjoint solution $R_{approx}$ of the reduced instance \((G-v,w')\), we have $R_{approx} \cup \{v\}$ is a $2$-approximate $p$-disjoint solution of the original instance $(G, w)$.
\end{lemma}
\begin{proof}

Since
  \((G-v)-R_{approx}=G-(R_{approx}\cup\{v\})\) and the former
  digraph is acyclic by assumption, we get that
  \(R_{approx}\cup\{v\}\) is a FVS in the digraph \(G\). We will show
  that  $R_{approx}\cup\{v\}$ is a $2$-approximate $p$-disjoint solution of $(G,w)$. 
Since $p\notin R_{approx}$, $R_{approx}\cup\{v\}$ is a $p$-disjoint FVS of $G$.  
Let $S^{\star}$ be an optimal $p$-disjoint solution of $(G,w)$. Notice that $S^{\star}\cap \{u,v\}\neq \emptyset$. 
Now to complete the proof, it  remains to show that \(w(R_{approx}\cup\{v\})\leq{}2w(S^{\star})\).
Let $\Delta=\min\{w(u),w(v)\}$, that is $w(v)=\Delta$. Now we have the following.
\begin{align*}
w(R_{approx}\cup\{v\})& = w'(R_{approx}\cup\{v\})+2\Delta  & \mbox{since }w(v)=\Delta \mbox{ and } w(u)=\Delta+w'(u)\\
&= w'(R_{approx})+2\Delta & \mbox{since }w'(v)=0\\
&\leq 2 w'(S^{\star}\setminus \{v\})+2\Delta   & \mbox{since $S^{\star}\setminus \{v\}$ is an FVS of $G-v$}\\
& = 2 w'(S^{\star})+2\Delta & \mbox{since }w'(v)=0\\
& = 2 \big(w(S^{\star}) -\Delta \cdot \vert S^{\star}\cap \{u,v\}\vert\big) +2\Delta  \\
&\leq 2 w(S^{\star}) & \mbox{since }  S^{\star}\cap \{u,v\}\neq \emptyset.
\end{align*}
This completes the proof. 
\end{proof}

Recall that in Phase~2 we work under the assumption that there is
an optimal solution \(S^{\star}\) of \((G,w)\) which does not
contain the pivot vertex \(p\).  If there is an arc $xy \in E(G)$
such that $x \in N^+(p) \setminus D_i$ and
$y \in N^-(p) \setminus D_i$ then the vertices \(\{x,p,y\}\) form
a triangle in \(G\), and so at least one of the two vertices
\(\{x,y\}\) must be present in the solution \(S^{\star}\). Let
\(v\) be a vertex of the least weight among \(\{x,y\}\), ties
broken arbitrarily, and let \(u\) be the other vertex.  Then
\autoref{lem:pickLighterVertex} applies to the tuple
\(\{(G,w),p,\{u,v\}\}\).

Procedure \({\sf Reduce}(G,w,p)\) of \autoref{alg:reduce}
\vpageref{alg:reduce} implements the reduction procedure of Phase
2. It starts by setting $D_0 = \emptyset$, $w_0 = w$, and $i = 0$.
As long as there is an arc $xy \in E(G)$ such that
$x \in N^+(p) \setminus D_i$ and $y \in N^-(p) \setminus D_i$ it
finds vertices \(\{u,v\}\) as described in the previous paragraph
and computes a weight function \(w'\) as specified in
\autoref{lem:pickLighterVertex} as applied to the collection
\(\{(G,w),p,\{u,v\}\}\). It sets $w_{i + 1} = w'$,
$D_{i+1} = D_i \cup \{v\}$, increments $i$ by one, and
repeats. When no such arc $xy$ exists the procedure outputs the
set $D = D_i$ and the weight function $\tilde{w} = w_i$.

\begin{algorithm}
  \caption{The reduction procedure of Phase 2.}\label{alg:reduce}
  \begin{algorithmic}[1]
    \Procedure{\sf Reduce}{\,\(G,w,p\)\,} 
      \State \(D_{0} \gets \emptyset\); \(w_{0} \gets w\); \(i \gets 0\)
      \While{\(G\) has an arc \(xy\;;\;x\in(N^{+}(p)\setminus{}D_{i}),y\in(N^{-}(p)\setminus{}D_{i})\)}\label{alg:reduce_while_start}
       \If{\(w_{i}(x) \leq w_{i}(y)\)}\label{alg:reduce_if_start} \Comment{definition of the vertices $u$ and $v$}
          \State \(v \gets x\);  \(u \gets y\)
        \Else
          \State \(v \gets y\); \(u \gets x\)
        \EndIf

        \State \(w_{i}(u) \gets w_{i}(u)-w_{i}(v)\)
        \State \(w_{i}(v) \gets 0\)
        \State \(w_{i+1} \gets w_{i}\)\label{alg:reduce_w_prime_update} \Comment{\(w_{i+1}\) is now the weight function $w'$ from the discussion}
        \State \(D_{i+1} \gets D_{i}\cup{}\{v\}\)\label{alg:reduce_D_update} 
        \State \(i \gets i+1\)
      \EndWhile\label{alg:reduce_while_end}

      \State \(D \gets D_{i}\); \(\tilde{w} \gets w_{i}\)\label{alg:reduce_final}
      \State \textbf{return} \((D,\tilde{w})\)    
    \EndProcedure
  \end{algorithmic}
\end{algorithm}

Our next lemma states that procedure {\sf Reduce} runs in
polynomial time and correctly outputs a reduced instance. Recall
that for an instance \((G,w)\) of \tfvs and a vertex
\(p\in{}V(G)\), a \emph{\(p\)-disjoint solution} of \((G,w)\) is
an FVS of \(G\) which does not contain vertex \(p\).

\begin{lemma}\label{lem:reduce}
  Let \((G,w)\) be an instance of \tfvs\  and $p\in V(G)$. 
When given $(G, w, p)$ as input, the procedure
  {\sf Reduce} runs in $O(|V(G)|^{2})$ time and outputs a
  vertex set $D \subseteq (V(G)\setminus{}\{p\})$ and a weight
  function $\tilde{w}$ with the following properties:
  \begin{itemize}
  \item[$(i)$] there are no arcs from $N^+(p)$ to $N^-(p)$ in digraph
    \(G-D\), and
    \item[$(ii)$] for every $2$-approximate $p$-disjoint solution $S$ of
     \((G-D, \tilde{w})\),  
      the set $S \cup D$ is a
      $2$-approximate $p$-disjoint solution of $(G, w)$.
  \end{itemize}
\end{lemma}
\begin{proof}

  The check on line~\ref{alg:reduce_while_start} of
  \autoref{alg:reduce} fails if and only if there are no arcs from
  $N^+(p)$ to $N^-(p)$ in the digraph \(G-D_{i}\) for the value of
  \(i\) at that point. Since the assignment of \(D_{i}\) to \(D\)
  on line~\ref{alg:reduce_final} happens only if this check fails,
  we get that there are no arcs from $N^+(p)$ to $N^-(p)$ in the
  digraph \(G-D\).
  Let $S$ be a $2$-approximate $p$-disjoint solution of 
  \((G-D, \tilde{w})\). 
Then by a simple induction on the number of iterations and \autoref{lem:pickLighterVertex}, we obtain that
$S\cup D$ is a $2$-approximate $p$-disjoint solution of $(G,w)$. 

To complete the proof we show that procedure \textsf{Reduce} runs
in \(O(n^{2})\) time where $n = |V(G)|$. Let
\(V(G)=\{v_{1},\dotsc,v_{n}\}\). We assume that graph \(G\) is
given as its \(n\times{}n\) adjacency matrix \(M_{G}\) where
\(M_{G}[i][j]=1\) if \(v_{i}v_{j}\) is an arc in \(G\) and
\(M_{G}[i][j]=0\) otherwise. We assume also that the weight
function \(w\) is given as a \(1\times{}n\) array where \(w[i]\)
stores the weight of vertex \(v_{i}\).

We compute the two neighborhoods \(N^{-}(p)\) and \(N^{+}(p)\) of
the pivot vertex \(p\) by scanning the entries of the row
\(M_{G}[p]\); vertex \(v_{i}\in{}N^{+}(p)\) if \(M_{G}[p][i]=1\),
and \(v_{i}\in{}N^{-}(p)\) if \(v_{i}\neq{}p\) and
\(M_{G}[p][i]=0\). This takes \(O(n)\) time. Let
\(d_{in}=|N^{-}(p)|,d_{out}=|N^{+}(p)|\) be the in- and
out-degrees of vertex \(p\). We construct a
\(d_{out}\times{}d_{in}\) array \AAA to store the neighborhood
relation between the sets \(N^{+}(p)\) and \(N^{-}(p)\), and a
\(1\times{}d_{out}\) array \(OD\) to store the out-degrees of
vertices in \(N^{+}(p)\) \emph{into the set} \(N^{-}(p)\). We
initialize all entries of \AAA and \(OD\) to zeroes. Now for each
pair of vertices \(v_{i}\in{}N^{+}(p),v_{j}\in{}N^{-}(p)\) we
increment the entries \(\AAA[i][j]\) and \(OD[i]\) by \(1\) each
if and only if \(M_{G}[i][j]=1\). Once this is done the cell
\(OD[i]\) holds the number of out-neighbors of vertex
\(v_{i}\in{}N^{+}(p)\) in the set \(N^{-}(p)\), and
\(\AAA[i][j]=1\) if and only if \(v_{i}v_{j}\) is an arc in \(G\)
for vertices \(v_{i}\in{}N^{+}(p),v_{j}\in{}N^{-}(p)\). Since
\(|N^{+}(p)|+|N^{-}(p)|=(n-1)\) all this can be done in
\(O(n^{2})\) time.

To execute the test on line~\ref{alg:reduce_while_start} of
\autoref{alg:reduce} we scan the list \(OD\) for a non-zero
entry. If all entries of \(OD\) are zeros then there is no arc
\(xy\) of the specified form and the test returns
\textbf{False}. If \(OD[i]>0\) for some \(i\) then we scan the row
\(\AAA[i]\) to find an index \(j\) such that
\(\AAA[i][j]=1\). Then \(x=v_{i},y=v_{j}\) is a pair of vertices
which satisfy the test. We use these vertices to execute
lines~\ref{alg:reduce_if_start} to~\ref{alg:reduce_w_prime_update}
of the procedure.  We effect the addition of vertex \(v\) to the
set \(D_{i+1}\) on line~\ref{alg:reduce_D_update} as follows: If
\(v=x=v_{i}\in{}N^{+}(p)\) then we set \(OD[i]=0\) and
\(\AAA[i][j]=0\;;\;1\leq{}j\leq{}d_{in}\). If
\(v=y=v_{j}\in{}N^{-}(p)\) then for each \(1\leq{}i\leq{}d_{out}\)
such that \(\AAA[i][j]=1\), we decrement the cells \(OD[i]\) and
\(\AAA[i][j]\) by \(1\).

Each line of \autoref{alg:reduce}, except for
line~\ref{alg:reduce_D_update}, takes constant
time. Line~\ref{alg:reduce_D_update}---as described above---takes
\(O(n)\) time. Each execution of line~\ref{alg:reduce_D_update}
takes either a row or a column of \AAA which has non-zero entries
and sets all these entries to zero. Since the algorithm does not
increment these entries in the loop, we get that the
\textbf{while} loop of lines~\ref{alg:reduce_while_start}
to~\ref{alg:reduce_while_end} is executed at most
\(|N^{+}(p)|+|N^{-}(p)|=(n-1)\) times. Thus the entire procedure
runs in \(O(n^{2})\) time.  
\end{proof}

Combining \autoref{lem:hereditarySolutions},
\autoref{lem:InOutNbrhoodSplit}, and \autoref{lem:reduce} we
get

\begin{corollary}\label{cor:reduceMain}
  On input $(G, w, p)$ the procedure {\sf Reduce} runs in $O(n^2)$
  time and outputs a vertex set $D \subseteq V(G) \setminus \{p \}$ and a  weight function $\tilde{w}$ such that for every FVS $S^-$ of $G[N^{-}(p) \setminus D]$ and every FVS $S^+$ of $G[N^{+}(p) \setminus D]$, we have that $S^- \cup S^+ \cup D$ is a  $p$-disjoint FVS of $G$. 
  
  Further, if $S^-$ is a $2$-approximate
  solution of $(G[N^{-}(p) \setminus D], \tilde{w})$ and $S^+$ is 
  $2$-approximate solution of $(G[N^{+}(p) \setminus D], \tilde{w})$ then $S^- \cup S^+ \cup D$ is a $2$-approximate $p$-disjoint solution of $(G, w)$.
\end{corollary}

\begin{proof}
The running time of procedure {\sf Reduce} follows from \autoref{lem:reduce}.  Let $S^-$ be an FVS of $G[N^{-}(p) \setminus D]$ and  $S^+$ be an FVS of $G[N^{+}(p) \setminus D]$. By \autoref{lem:reduce}, there are no arcs from $N^+(p)$ to $N^-(p)$ in digraph \(G-D\). Then  by statement $(i)$ of \autoref{lem:InOutNbrhoodSplit}, 
$p$ is not part of any triangle in $G-D$. Thus, by statement $(ii)$ of \autoref{lem:InOutNbrhoodSplit}, $S^-\cup S^+$ is an FVS of $G-D$. Therefore, 
by \autoref{lem:hereditarySolutions}, $S^- \cup S^+ \cup D$ is an FVS of $G$. Moreover, since $p\notin S^- \cup S^+ \cup D$, it is a $p$-disjoint FVS of $G$.  

Suppose $S^-$ is a $2$-approximate solution of $(G[N^{-}(p) \setminus D], \tilde{w})$ and $S^+$ is a
  $2$-approximate solution of $(G[N^{+}(p) \setminus D], \tilde{w})$. Now we claim that $S^-\cup S^+$ is a $2$-approximate $p$-disjoint solution of $(G-D,\tilde{w})$. 
Let \(R^{-}\) and \(R^{+}\) be optimal solutions  of $(G[N^{-}(p) \setminus D], \tilde{w})$ and  $(G[N^{+}(p) \setminus D], \tilde{w})$, respectively. Then we claim that \(R^{-}\cup R^{+}\) is an optimal $p$-disjoint solution of $(G-D, \tilde{w})$. By statement $(ii)$ of \autoref{lem:InOutNbrhoodSplit}, $R^-\cup R^+$ is an FVS of $G-D$ and clearly it does not contain $p$. Suppose $R^-\cup R^+$ is not an optimal $p$-disjoint solution of $(G-D, \tilde{w})$. Let $R^{\star}$ be an optimal $p$-disjoint solution 
of $(G-D, \tilde{w})$ and $\tilde{w}(R^{\star})<\tilde{w}(R^{-}\cup R^{+})$. Then, either $\tilde{w}(R^{\star}\cap (N^{-}(p) \setminus D))<\tilde{w}(R^{-})$ or $\tilde{w}(R^{\star}\cap (N^{+}(p) \setminus D))<\tilde{w}(R^{+})$. Consider the case when $\tilde{w}(R^{\star}\cap (N^{-}(p) \setminus D))<\tilde{w}(R^{-})$. By \autoref{lem:InOutNbrhoodSplit}, $R^{\star}\cap (N^{-}(p) \setminus D)$ is an FVS of $G[N^{+}(p) \setminus D]$. But this contradicts the assumption that \(R^{-}\) is an optimal solution of $(G[N^{-}(p) \setminus D], \tilde{w})$. The same arguments apply to the case when 
$\tilde{w}(R^{\star}\cap (N^{+}(p) \setminus D))<\tilde{w}(R^{+})$. Therefore \(R^{-}\cup R^{+}\) is an optimal $p$-disjoint solution of $(G-D, \tilde{w})$. Since $S^-$ is a $2$-approximate solution of $(G[N^{-}(p) \setminus D], \tilde{w})$ and $S^+$ is a $2$-approximate solution of $(G[N^{+}(p) \setminus D], \tilde{w})$, we have that $\tilde{w}(S^-\cup S^+)=\tilde{w}(S^-)+ \tilde{w}(S^+)\leq 2(\tilde{w}(R^-)+\tilde{w}(R^+))\leq 2\tilde{w}(R^{-}\cup R^{+})$. Hence, $S^-\cup S^+$ is a $2$-approximate $p$-disjoint solution of $(G-D,\tilde{w})$. 
Then by \autoref{lem:reduce}, $S^- \cup S^+ \cup D$ is a $2$-approximate $p$-disjoint solution of $(G, w)$. This completes the proof of the corollary. 
\end{proof}

We are now ready to prove our main theorem. 
\mainthm*

\begin{proof}
We first describe the algorithm. On input $(G, w)$, if $G$ has at most $10$ vertices the algorithm finds an optimal solution by exhaustively enumerating and comparing all potential solutions. Otherwise the algorithm iteratively computes at most \(26\)
solutions of $(G,w)$ by making recursive calls. It then outputs the least weight FVS among them. We now describe the iterations and the recursive calls. Let us index the iteration by $i \in \{0,1, \ldots,25\}$.

The first iteration is different from the other \(25\) iterations. In this iteration, the algorithm sets $D \subseteq V(G)$ to be the set of the $\frac{n}{6}$ vertices of smallest weight in $V(G)$ and  $\Delta = \max_{v \in D} w(v)$. 
Let \(w'\colon V(G)\setminus D \rightarrow {\mathbb N}\) be the weight function 
  which assigns the weight \(w(v)-\Delta\) to each vertex \(v\) of
  \(G-D\). 
The algorithm calls itself recursively on $(G - D, w')$. The recursive call returns an FVS $S$ of $G - D$, the algorithm constructs the FVS $S_0 = S \cup D$ of $G$.

We do the remaining 25 iterations
only when the set $\{v : N^+(v) \leq 8n/9, N^-(v) \leq 8n/9\}$ is non-empty. 
For each of these 25 iterations (which we index by $i \in \{1,2, \ldots,25\}$), the algorithm picks a vertex $p_i$ uniformly at random from the set of vertices $\{v : N^+(v) \leq 8n/9, N^-(v) \leq 8n/9\}$. 
For each $p_i$ the algorithm runs the procedure {\sf Reduce} on $G$, $p_i$, and $w$ and obtains a set $D_i$ and a weight function $\tilde{w}_i$. It then makes two recursive calls, one on $(G[N^-(p_i) \setminus D_i], \tilde{w}_i)$, and the other on $(G[N^+(p_i) \setminus D_i], \tilde{w}_i)$. Let the sets returned by the two recursive calls be $S^-_i$ and $S^+_i$ respectively. The algorithm constructs the set $S_i = S^-_i \cup S^+_i \cup D_i$ as the FVS of G corresponding to $i$. 

Finally, the algorithm outputs the minimum weight $S_i$, where the minimum is taken over $0 \leq i \leq 25$ as the solution. The algorithm terminates within the claimed  running time, since the running time is governed by the recurrence $T(n) \leq 51 \cdot T(8n/9) + O(n^2)$ which solves to $T(n) = O(n^{34})$ by the Master theorem~\cite{Cormen:2009:IAT:1614191}. 
We now prove that in each iteration, the constructed solution $S_i$ is indeed an FVS of $G$, and that the same holds for the solution returned by the algorithm. We apply an induction on the number of vertices in $G$. For $n \leq 10$ there are no recursive calls made, and the returned solution is an optimal solution, since it is computed by brute force. For $n > 10$ the returned solution is one of the $S_i$'s and so it is sufficient to prove that all $S_i$'s are in fact feedback vertex sets of $G$.
For $S_i$, $i \geq 1$ this follows from Corollary~\ref{cor:reduceMain} and the induction hypothesis. 
And for $i=0$, we know that $S_0=S\cup D$ and $S$ is a vertex subset returned by the recursive call for the instance $(G-D,w')$, which is also an FVS of $G-D$, by the induction hypothesis.  Since $G-S_0=((G-D)-S)$ and $S$ is an FVS of $(G-D)$, clearly $S_0$ is an FVS of $G$.  

Finally, will show that with probability at least $1/2$, the algorithm outputs a $2$-approximate solution of $(G,w)$. We prove this by induction on $n$, the number of vertices in $G$. Suppose that $S_i$ is of the least weight among $S_0,S_1, \ldots, S_{25}$, for some $i \in \{0,2,\ldots 25\}$, which is output by the algorithm. For $n \leq 10$ the returned solution is optimal, so assume $n > 10$. Let $S_{OPT}$ be an optimal solution for $(G, w)$. We distinguish between two cases, either $|S_{OPT}| \geq 2n/3$ or $|S_{OPT}| < 2n/3$.
If $|S_{OPT}| \geq 2n/3$ then, by the induction hypothesis the first iteration, the recursive call on $(G - D, w')$ returns a $2$-approximate solution $S$ for $(G - D, w')$ with probability at least $1/2$. In this case it follows from \autoref{lem:largeSoln} that $S_i$ for $i=0$, is a $2$-approximate solution for $(G, w)$.

Suppose now that $|S_{OPT}| < 2n/3$. We will argue that in each of the $25$ remaining iterations the probability that $p_i \notin S_{OPT}$ is at least $1/9$. Indeed, $G - S_{OPT}$ is an acyclic tournament on at least $n/3$ vertices. Let $R$ be the set of vertices in $V(G) \setminus S_{OPT}$ excluding the first $\lfloor n/9 \rfloor$ vertices and the last $\lfloor n/9 \rfloor$ vertices in the unique topological order of the acyclic tournament $G - S_{OPT}$. For each vertex $v$ in $R$ it holds that $|N^+(v)| \leq n - \lfloor n/9 \rfloor - 1 \leq 8n/9$ and similarly $|N^-(v)| \leq 8n/9$, i.e. $R \subseteq \{v : N^+(v) \leq 8n/9, N^-(v) \leq 8n/9\}$.
Furthermore, $|R| \geq n/9$ since $|V(G) \setminus S_{OPT}| \geq n/3$. Hence, when we pick a random vertex $p_i$ among all vertices with in-degree and out-degree at most $8n/9$ we have that with probability at least $1/9$ the vertex $p_i$ is in $R$, and therefore not in $S_{OPT}$.

We shall say that an iteration $i$ with $i \geq 1$ is {\em good} if $p_i \notin S_{OPT}$ and the two solutions $S^-_i$ and $S^+_i$ returned from the recursive calls on $(G[N^-(p_i) \setminus D_i], \tilde{w}_i)$ and $(G[N^+(p_i) \setminus D_i], \tilde{w}_i)$, respectively are $2$-approximate for their respective instances. Since $p_i \notin S_{OPT}$ with probability at least $1/9$, and each of $S^-_i$ and $S^+_i$ are $2$-approximate with probability at least $1/2$ (by the induction hypothesis), it follows that this iteration is good with probability at least $1/9 \cdot 1/2 \cdot 1/2 \geq 1/36$. 
Therefore, with probability at least
$$1 - (1 - 1/36)^{25} \geq 1/2$$
there is at least one iteration $i$ which is good. 
For this iteration it follows from Corollary~\ref{cor:reduceMain} that $S_i = D_i \cup S^+_i \cup S^-_i$ is $2$-approximate $p_i$-disjoint solution of $(G,w)$. Moreover, since $p_i \notin S_{OPT}$, $S_{OPT}$ is also an optimal $p_i$-disjoint solution of $(G,w)$.   
Hence $w(S_i)\leq 2w(S_{OPT})$. 
Therefore the solution output by the algorithm is a $2$-approximate solution with probability at least $1/2$.
This concludes the proof.
\end{proof}

\subsection{Deterministic $2$-approximation in quasi-polynomial time.}
We can easily derandomize the above algorithm in quasi-polynomial
time.  Instead of randomly selecting the pivots $p_i$, we iterate
over all the candidates in
$\{v : N^+(v) \leq 8n/9, N^-(v) \leq 8n/9\}$.  The correctness of
this algorithm follows from the same arguments as above, and we
obtain a deterministic $2$-approximation algorithm for \tfvs.  To
bound the running time, observe that the number of recursive calls
will be at most $2n+1$. Thus the running time of the algorithm
will be governed by the recurrence
$T(n) \leq (2n+1) \cdot T(8n/9) + O(n^2)$ which solves to
$T(n) = n^{O(\log n)}$ by the Master
theorem~\cite{Cormen:2009:IAT:1614191}.  Thus we get the following
theorem.

\begin{restatable}{thm}{quasithm}
\label{thm:quasi}
There exists an algorithm that given an instance $(G,w)$ of \tfvs
on \(n\) vertices, runs in time $n^{O(\log n)}$ and outputs a
$2$-approximate solution of $(G,w)$.
\end{restatable}

%
 
\section{Conclusions}\label{conclusion}
We presented a simple randomized $2$-approximation algorithm for {\sc Feedback Vertex Set in Tournaments}.
Assuming the Unique Games conjecture, the approximation ratio is optimal. However there is still some room for improvement. First and foremost, is it possible to obtain a deterministic $2$-approximation algorithm? Further, for the sake of clarity of presentation we did not attempt at all to optimize the running time of the algorithm. The exponent $34$ can be brought down substantially by implementing the following. 
\begin{enumerate}\setlength\itemsep{-.7mm}
\item Changing the threshold $2n/3$ for when $|S_{OPT}|$ is considered big (and the first recursive call returns an optimal solution) to $\alpha n$. In this case the set $D$ must be chosen to be the set of $(\alpha - 1/2)n$ vertices of smallest weight. 

\item Changing the success probability with which the algorithm
  returns a solution from $1/2$ to some constant $r$. This allows
  to reduce the number of iterations.

\item Changing the maximum indegree and outdegree of the sampled vertices $p_i$ from $8n/9$ to $\beta n$. This gives a trade-off between the probability that each iteration is good, and the upper bound on the size of the digraphs $G[N^-(p_i) \setminus D_i]$ and $G[N^+(p_i) \setminus D_i]$ in the recursive calls.

\item Instead of computing the probability that the pivot $p_i$ is in $R$, computing the probability that $p_i$ is not in $S_{OPT}$. In particular vertices in $V(G) \setminus (S_{OPT} \cup  R)$ either have both indegree and outdegree at most $\lfloor 8n/9 \rfloor$, in which case they contribute equally to the numerator and the denominator of the probability, or they do not, in which case they contribute to neither the numerator nor the denominator. The worst probability is achieved in the latter case, making the probability that $p_i \notin S_{OPT}$ be at least $1/7$ (instead of the lower bound of $1/9$ of being in $R$).

\item\label{pt:assymmetry} Not using the same upper bound on the number of vertices in all recursive calls. The first recursive call is made on an instance with (potentially) fewer vertices. More importantly, in each of the remaining iterations the algorithm makes two recursive calls, one with $\gamma_i n$ vertices and the other with $(1 - \gamma_i) n$ vertices. In our analysis we just used that $\gamma_i \leq 8/9$ and $(1 - \gamma_i) \leq 8/9$ without also using that in the worst case when $\gamma_i = 8/9$ we have $1 - \gamma_i = 1/9$.

\item Taking point~\ref{pt:assymmetry} one step further, after the algorithm has sampled $p_i$ it can observe what $\gamma_i$ is. It may then make several recursive calls on $G[N^-(p_i) \setminus D_i]$ and on $G[N^+(p_i) \setminus D_i]$, this gives another tradeoff between the time spent and the success probability that a particular iteration is good. Note that the number of recursive calls on  $G[N^-(p_i) \setminus D_i]$ and on $G[N^+(p_i) \setminus D_i]$ need not be the same - indeed it pays off to make more recursive call to the smaller instance, since that provides the best trade-off between running time and success probability. In particular the number of calls on  $G[N^-(p_i) \setminus D_i]$ and on $G[N^+(p_i) \setminus D_i]$ should be chosen as a function of $\gamma_i$.
\end{enumerate}
Nevertheless this is still a far cry from a practical running time, and it would be interesting to see whether one can achieve the same approximation ratio can be obtiained by an algorithm with a running time of $O(n^2)$ (i.e. linear in input size) or something close.

Finally it would be interesting to see whether ideas from this
algorithm can be used to improve approximation algorithms for
other  ``structured hitting-set'' problems. Here the {\sc Cluster
  Vertex Deletion} problem is a possible candidate.

\bibliography{fvst}{}

\end{document}